\documentclass[submission,copyright,creativecommons]{eptcs}
\providecommand{\event}{} 
\usepackage{underscore}

\usepackage{amssymb}

\def\ra{\rightarrow}
\def\lra{\longrightarrow}
\def\nat{\mbox{\it nat}}
\def\bool{\mbox{\it bool}}
\def\true{\mbox{\it true}}
\def\array{\mbox{\it array}}
\def\nil{\mbox{\it nil}}

\newbox\tempa
\newbox\tempb
\newdimen\tempc
\def\mud#1{\hfil $\displaystyle{\mathstrut #1}$\hfil}
\def\rig#1{\hfil $\displaystyle{#1}$}
\def\irulehelp#1#2#3{\setbox\tempa=\hbox{$\displaystyle{\mathstrut #2}$}%
		        \setbox\tempb=\vbox{\halign{##\cr
	\mud{#1}\cr
\noalign{\vskip\the\lineskip}%
	\noalign{\hrule height 0pt}%
	\rig{\vbox to 0pt{\vss\hbox to 0pt{${\; #3}$\hss}\vss}}\cr
	\noalign{\hrule}%
	\noalign{\vskip\the\lineskip}%
	\mud{\copy\tempa}\cr}}%
		      \tempc=\wd\tempb
		      \advance\tempc by \wd\tempa
		      \divide\tempc by 2 }
\def\irule#1#2#3{{\irulehelp{#1}{#2}{#3}%
		     \hbox to \wd\tempa{\hss \box\tempb \hss}}}

\newtheorem{definition}{Definition}[section]
\newtheorem{theorem}{Theorem}[section]
\newtheorem{lemma}{Lemma}[section]
\newenvironment{proof}{\noindent {\em Proof.}}{\medskip}
\newenvironment{example}{\noindent {\em Example.}}{\medskip}

\title{Interacting Safely with an Unsafe Environment}

\author{Gilles Dowek
\institute{Inria and ENS Paris-Saclay}
\email{gilles.dowek@ens-paris-saclay.fr}}

\def\titlerunning{Interacting Safely with an Unsafe Environment}

\def\authorrunning{Gilles Dowek}

\begin{document}
\maketitle

\begin{abstract}
We give a presentation of Pure type systems where contexts need not be
well-formed and show that this presentation is equivalent to the usual
one.  The main motivation for this presentation is that, when we
extend Pure type systems with computation rules, like in the logical
framework {\sc Dedukti}, we want to declare the constants before the
computation rules that are needed to check the well-typedness of their
type.
\end{abstract}

\section{Introduction}

In the simply typed lambda-calculus, to assign a type to a term, we
first need to assign a type to its free variables. For instance, if we
assign the type $\nat \ra \nat \ra \nat$ to the variable $f$ and the
type $\nat$ to the variable $x$, then we can assign the type $\nat \ra
\nat$ to the term $\lambda y:\nat~(f~x~y)$.

Whether a type is assigned to $f$ before or after one is assigned to
$x$ is immaterial, so the context $\{f:\nat \ra \nat \ra \nat,
x:\nat\}$ does not need to be ordered.

\subsection{Well-formed Contexts}

In systems, such as the Calculus of constructions, where atomic types
are variables of a special type $*$, contexts are ordered and, for
instance, the term $\lambda y:\nat~(f~x~y)$ is assigned the type $\nat
\ra \nat$ in the context $\nat:*, f:\nat \ra \nat \ra \nat, x:\nat$
but not in the context $x:\nat, \nat:*, f:\nat \ra \nat \ra \nat$,
that is not well-formed.

In a well-formed context, the declarations are ordered in such a way
that the type of a variable only contains variables declared to its
left.  For instance, the context $\nat:*, z:\nat, \array : \nat \ra *,
\nil:(\array~z)$ is well-formed, but the context $\nat:*, z:\nat,
\nil:(\array~z), \array : \nat \ra *$ is not.  Moreover, in such a
well-formed context, each type is itself well-typed in the context
formed with the variable declarations to its left.  For instance, the
context $\nat:*, \array : \nat \ra *, z:\nat, \nil:(\array~z~z)$ is
not well-formed.  So, a context $x_1:A_1, ..., x_n:A_n$ is said to be
well-formed if, for each $i$, $x_1:A_1, ..., x_i:A_i \vdash A_{i+1}:s$
is derivable for some sort $s$ in $\{*, \Box\}$.

The original formulation of the Calculus of constructions of Coquand
and Huet \cite{CoquandHuet} has two forms of judgements: one
expressing that a context $\Gamma$ is well-formed and another
expressing that a term $t$ has a type $A$ in a context $\Gamma$.  Two
rules define when a context is well-formed
$$\irule{}
        {[~]~\mbox{well-formed}}
        {\mbox{(empty)}}$$
$$\irule{\Gamma \vdash A:s} {\Gamma, x:A~\mbox{well-formed}}
        {\mbox{(decl) $s \in \{*,\Box\}$}}$$
and one enables the assignment of a type to a variable, in a
well-formed context
$$\irule{\Gamma, x:A, \Gamma'~\mbox{well-formed}}
        {\Gamma, x:A, \Gamma' \vdash x:A}
        {\mbox{(var)}}$$
These three rules together with five others---(sort), (prod), (abs),
(app), and (conv)---form a eight-rule presentation of the Calculus of
constructions, and more generally of Pure type systems.  Because the
rule (var) requires the context $\Gamma, x:A, \Gamma'$ to be
well-formed, a variable can only be assigned a type in a well-formed
context and this property extends to all terms, as it is an invariant
of the typing rules.

This system was simplified by Geuvers and Nederhof
\cite{GeuversNederhof} and Barendregt \cite{Barendregt1992}, who use a
single form of judgement expressing that a term $t$ has a type $A$ in
a context $\Gamma$.  First, they drop the context $\Gamma'$ in the
rule (var) simplifying it to
$$\irule{\Gamma, x:A~\mbox{well-formed}}
        {\Gamma, x:A \vdash x:A}
        {}$$
and add a weakening rule
$$\irule{\Gamma \vdash t:A~~~\Gamma \vdash B:s}
        {\Gamma, x:B \vdash t:A}
        {\mbox{(weak)}}$$
to extend the judgement $\Gamma, x:A \vdash x:A$ to $\Gamma, x:A,
\Gamma' \vdash x:A$. Then, they exploit the fact that the conclusion
of the rule (decl) is now identical to the premise of the rule (var), to
coin a derived rule
$$\irule{\Gamma \vdash A:s} {\Gamma, x:A \vdash x:A} {\mbox{(start) $s
    \in \{*,\Box\}$}}$$ Now that the variables can be typed without
using a judgement of the form $\Gamma~\mbox{well-formed}$, such
judgements can be dropped, together with the rules (empty), (decl),
and (var). So the two rules (start) and (weak), together with the
five other rules form an equivalent seven-rule formulation of the Calculus of
constructions, and more generally of Pure type systems.

\subsection{Interacting Safely with an Unsafe Environment}
           
When a judgement of the form $\Gamma \vdash x:A$ is derived, the
well-typedness of the term $A$ needs to be checked. But it can be
checked either when the variable $x$ is added to the context or when
it is used in the derivation of the judgement $\Gamma \vdash x:A$.  In
the system with the rules (decl) and (var), it is checked in the rule
(decl), that is when the variable is added to the context. When the
rule (var) is replaced with the rule (start), it is checked when the
variable is used. These two systems illustrate two approaches to
safety: the first is to build a safe environment, the second is to
interact safely with a possibly unsafe environment.

In the formulation of Geuvers and Nederhof and Barendregt, it is still
possible to define a notion of well-formed context: the context
$x_1:A_1, ..., x_n:A_n$ is well-formed if for each $i$, $x_1:A_1, ...,
x_i:A_i \vdash A_{i+1}:s$ is derivable.  With such a definition, it is
possible to prove that if the judgement $\Gamma \vdash t:A$ is
derivable, then $\Gamma$ is well-formed. In this proof, the second
premise of the rule (weak), $\Gamma \vdash B:s$, is instrumental, as
its only purpose is to preserve the well-formedness of the context.

We can go further with the idea of interacting safely with an unsafe
environment and drop this second premise, leading to the weakening rule
$$\irule{\Gamma \vdash t:A}
        {\Gamma, x:B \vdash t:A}
        {}$$
Then, in the judgement $\Gamma \vdash x:A$,
nothing prevents the context $\Gamma$ from being non well-formed, but
the term $A$ is still well-typed because the rule (start), unlike the
rule (var), has a premise $\Gamma \vdash A:s$. In such a system, the
judgement $\nat:*, \array : \nat \ra *, z:\nat, \nil:(\array~z~z)
\vdash z:\nat$ is derivable, although the term $(array~z~z)$ is not
well-typed, but the judgement $\nat:*, \array : \nat \ra *, z:\nat,
\nil:(\array~z~z) \vdash \nil:(\array~z~z)$ is not because this term
$(array~z~z)$ is not well-typed.

Yet, with the rule (start) and this strong weakening rule, the
judgement $\nat:*, z:\nat, \nil:(\array~z), 
\linebreak
\array : \nat \ra * \vdash
\nil:(\array~z)$ is not derivable, because the judgement $\nat:*,
z:\nat \vdash (\array~z):*$ is not derivable. Thus, to make this
judgement derivable, we should not use a weakening rule that erases
all the declarations to the right of the declaration of $\nil$ and the
rule (start). But we should instead use a rule that keeps the full
context to type the term $(\array~z)$. Yet, like the rule (start),
this rule should not check that the context is well-formed, but that
the type of the variable is a well-typed term
$$\irule{\Gamma,x:A,\Gamma' \vdash A:s}
        {\Gamma,x:A,\Gamma' \vdash x:A}
        {\mbox{(var')}}$$
\begin{figure}
$$\irule{}
        {\Gamma \vdash s_1:s_2}
        {\mbox{(sort')}~~~\langle s_1,s_2 \rangle \in {\cal A}}$$
$$\irule{\Gamma,x:A,\Gamma' \vdash A:s}
        {\Gamma,x:A,\Gamma' \vdash x:A}
        {\mbox{(var')}~~~x \in {\cal V}_s}$$
$$\irule{\Gamma \vdash A:s_1~~~\Gamma, x:A \vdash B:s_2}
        {\Gamma \vdash (x:A) \ra B:s_3}
        {\mbox{(prod)}~~~\langle s_1,s_2,s_3 \rangle \in {\cal R}}$$
$$\irule{\Gamma \vdash A:s_1~~~\Gamma, x:A \vdash B:s_2~~~\Gamma, x:A \vdash t:B}
        {\Gamma \vdash \lambda x:A~t:(x:A) \ra B}
        {\mbox{(abs)}~~~\langle s_1,s_2,s_3 \rangle \in {\cal R}}$$
$$\irule{\Gamma \vdash t:(x:A) \ra B~~~\Gamma \vdash u:A}
        {\Gamma \vdash t~u:(u/x)B}
        {\mbox{(app)}}$$
$$\irule{\Gamma \vdash t:A~~~\Gamma \vdash B:s}
        {\Gamma \vdash t:B}
        {\mbox{(conv)}~~~A \equiv B}$$
\caption{Pure type systems with arbitrary contexts \label{typingrules}}
\end{figure}
This leads to the six-rule system described in Figure \ref{typingrules}.

As the order of declarations in a context is now immaterial, contexts
can indifferently be defined as sequences or as sets of declarations.

\subsection{Previous Work}

There are several reasons for using arbitrary contexts.  One of them
is that, as already noticed by Sacerdoti Coen \cite{Sacerdoti}, when
we have two contexts $\Gamma$ and $\Gamma'$, for instance developed by
different teams in different places, and we want to merge them, we
should not have to make a choice between $\Gamma, \Gamma'$, and
$\Gamma', \Gamma$. We should just be able to consider the unordered
context $\Gamma \cup \Gamma'$, provided it is a context, that is if
$x:A$ is declared in $\Gamma$ and $x:A'$ is declared in $\Gamma'$ then
$A = A'$.

Another is that, when we extend Pure type systems with computation
rules, like in the logical framework {\sc Dedukti}, we additionally
want to declare constants in a signature $\Sigma$ and then add
computation rules. For instance, we want to be able to declare
constants in a signature $\Sigma = \nat:*,a:\nat, b:\nat, P:\nat \ra
*, Q:(P~a) \ra *, e:(P~b), h:(Q~e), c:\nat$ and then computation rules
$a \lra c$, $b \lra c$.  Because, unlike in \cite{CD}, the term
$(Q~e)$ is not well-typed without the computation rules, we cannot
check the that the signature is well-formed before we declare the
rules. But, because the rules use the constants declared in $\Sigma$,
we cannot declare the rules before the signature, in particular the
rules do not make sense in the part of the signature to the left of
the declaration of $h$, that is in $\nat:*,a:\nat, b:\nat, P:\nat \ra
*, Q:(P~a) \ra *, e:(P~b)$, where the constant $c$ is missing. And,
because we sometimes want to consider rules $l \lra r$ where $l$ and
$r$ are not well-typed terms \cite{Blanqui}, we cannot interleave
constant declarations and computation rules. Note that in Blanqui's
Calculus of algebraic constructions \cite{Blanqui}, the contexts are
required to be well-formed, but the signatures are not.

Another source of inspiration is the presentation of Pure type systems
without explicit contexts \cite{GeuversKrebbersMcKinnaWiedijk}, where
Geuvers, Krebbers, McKinna, and Wiedijk completely drop contexts in
the presentation of Pure type systems.  In particular, Theorem
\ref{maintheorem} below is similar to their Theorem 19.  The
presentation of Figure \ref{typingrules} is however milder than their
Pure type systems without explicit contexts, as it does not change the
syntax of terms, avoiding, for instance, 
the question of the convertibility of $x^B$ and $x^{(\lambda
  \dot{A}:*\dot{A})~B}$. In particular, if $\Gamma \vdash t:A$ is
derivable in the usual formulation of Pure type systems, it is also
derivable in the system of Figure \ref{typingrules}.

We show, in this note, that the system presented in Figure
\ref{typingrules} indeed allows to interact safely with an unsafe
environment, in the sense that if a judgement $\Gamma \vdash t:A$ is
derivable in this system, then there exists $\Delta$, such that
$\Delta \subseteq \Gamma$ and $\Delta \vdash t:A$ is derivable with
the usual Pure type system rules. The intuition is that, because of
the rule (var'), the structure of a derivation tree induces a
dependency between the used variables of $\Gamma$ that is a partial
order, and as already noticed by Sacerdoti Coen \cite{Sacerdoti}, a
topological sorting of the used variables yields a linear context
$\Delta$.  Topological sorting is the key of Lemma \ref{curation}.

So this paper build upon the work of Coquand and Huet
\cite{CoquandHuet}, Geuvers and Nederhof \cite{GeuversNederhof},
Barendregt \cite{Barendregt1992}, Blanqui \cite{Blanqui}, Sacerdoti Coen
\cite{Sacerdoti}, and Geuvers, Krebbers, McKinna, and Wiedijk
\cite{GeuversKrebbersMcKinnaWiedijk}. Its main contribution is to show
that Pure type systems can be defined with six rules only, without a
primitive notion of well-formed context, and without changing the
syntax of terms.

\section{Pure Type Systems}

Let us first recall a usual definition of (functional) Pure type
systems \cite{GeuversNederhof,Barendregt1992}.  To define the syntax
of terms, we consider a set ${\cal S}$ of sorts and a family of ${\cal
  V}_s$ of infinite and disjoint sets of variables of sort $s$. The
syntax is then
$$t = x~|~s~|~(x:A) \ra B~|~\lambda x:A~t~|~t~u$$

A context $\Gamma$ is a sequence $x_1:A_1, ..., x_n:A_n$ of pairs
formed with a variable and a term, such that the variables $x_1, ...,
x_n$ are distinct.  So, when we write the context $\Gamma, y:B$, we
implicitly assume that $y$ is not already declared in $\Gamma$.

A context $\Gamma$ is said to be included into a context $\Gamma'$
($\Gamma \subseteq \Gamma'$) if every $x:A$ in $\Gamma$ is also in
$\Gamma'$.

Two contexts $\Gamma$ and $\Gamma'$ are said to be compatible if each
time $x:A$ is in $\Gamma$ and $x:A'$ is in $\Gamma'$, then $A = A'$.

To define the typing rule, we consider a set ${\cal A}$ of axioms,
that are pairs of sorts and a set ${\cal R}$ of rules, that are triple
of sorts.  As we restrict to functional Pure type systems, we assume
that the relations ${\cal A}$ and ${\cal R}$ are functional.

\begin{definition}[The type system ${\cal T}$]
$$\irule{}
        {\vdash s_1:s_2}
        {\mbox{(sort)}~~~\langle s_1,s_2 \rangle \in {\cal A}}$$
$$\irule{\Gamma \vdash A:s}
        {\Gamma, x:A \vdash x:A}
        {\mbox{(start)}~~~x\in {\cal V}_s}$$
$$\irule{\Gamma \vdash t:A~~~\Gamma \vdash B:s}
        {\Gamma, x:B \vdash t:A}
        {\mbox{(weak)}~~~x\in {\cal V}_s}$$
$$\irule{\Gamma \vdash A:s_1~~~\Gamma, x:A \vdash B:s_2}
        {\Gamma \vdash (x:A) \ra B:s_3}
        {\mbox{(prod)}~~~\langle s_1,s_2,s_3 \rangle \in {\cal R}}$$
$$\irule{\Gamma \vdash A:s_1~~~\Gamma, x:A \vdash B:s_2~~~\Gamma, x:A \vdash t:B}
        {\Gamma \vdash \lambda x:A~t:(x:A) \ra B}
        {\mbox{(abs)}~~~\langle s_1,s_2,s_3 \rangle \in {\cal R}}$$
$$\irule{\Gamma \vdash t:(x:A) \ra B~~~\Gamma \vdash u:A}
        {\Gamma \vdash t~u:(u/x)B}
        {\mbox{(app)}}$$
$$\irule{\Gamma \vdash t:A~~~\Gamma \vdash B:s}
        {\Gamma \vdash t:B}
        {\mbox{(conv)}~~~A \equiv B}$$
\end{definition}

\begin{example}
  Consider two sorts $*$ and $\Box$ and an axiom $*:\Box$.  The
  judgement $\nat:*, z:\nat \vdash z:\nat$ is derivable in ${\cal
    T}$. But the judgements $z:\nat, \nat:* \vdash z:\nat$ is not
  because $z$ is declared before $\nat$ and the judgement $\nat:*,
  x:(*~*), z:\nat \vdash z:\nat$ is not because $(*~*)$ is not
  well-typed.
\end{example}

\begin{definition}[Well-formed]
Well-formed contexts are inductively defined with the rules 
\begin{itemize}
\item the empty context is well-formed,
\item if $\Gamma$ is well-formed and $\Gamma \vdash A:s$ is derivable in ${\cal
  T}$, then $\Gamma, x:A$ is well-formed.
\end{itemize}
\end{definition}

\begin{lemma}
  If $\Gamma \vdash t:A$ is derivable, then $\Gamma$ is well-formed.
  Conversely, if $\Gamma$ is well-formed, then there exist two terms
  $t$ and $A$, such that $\Gamma \vdash t:A$ is derivable.
\end{lemma}

\begin{proof}
We prove that $\Gamma$ is well-formed, by induction on the derivation
of $\Gamma \vdash t:A$. Conversely, if $\Gamma$ is well-formed and
$s_1$ and $s_2$ are two sorts, such that $\langle s_1, s_2 \rangle \in
{\cal A}$ then $\Gamma \vdash s_1:s_2$ is derivable with the rules
(sort) and (weak).
\end{proof}

We will use the two following lemmas.
The first is Lemma 18 in \cite{GeuversNederhof} and 5.2.12 in
\cite{Barendregt1992} and the second Lemma 26 in
\cite{GeuversNederhof} and 5.2.17 in \cite{Barendregt1992}.

\begin{lemma}[Thinning]
\label{thinningT}
If $\Gamma \vdash t:A$ is derivable, $\Gamma \subseteq \Gamma'$, and
$\Gamma'$ is well-formed, then $\Gamma' \vdash t:A$ is derivable.
\end{lemma}

\begin{lemma}[Strengthening]\label{strengthening}
  If $\Gamma, x:A, \Gamma' \vdash t:B$ is derivable and $x$ does not
  occur in $\Gamma'$, $t$, and $A$, then $\Gamma, \Gamma' \vdash t:B$
  is derivable.
\end{lemma}

\begin{lemma}[Strengthening contexts]\label{strengtheningcontexts}
If $\Gamma_1, x:A, \Gamma_2$ is well-formed and $x$ does not occur in
$\Gamma_2$ then $\Gamma_1, \Gamma_2$ is well-formed.
\end{lemma}

\begin{proof}
By induction
on the structure of $\Gamma_2$.  If $\Gamma_2$ is empty, then
$\Gamma_1, \Gamma_2 = \Gamma_1$ is well-formed.  Otherwise, $\Gamma_2
= \Gamma'_2, y:B$.  By induction hypothesis, $\Gamma_1, \Gamma'_2$ is
well-formed.  As $\Gamma_1, x:A, \Gamma'_2, y:B$ is well-formed,
$\Gamma_1, x:A, \Gamma'_2 \vdash B:s$ is derivable.  By Lemma
\ref{strengthening}, $\Gamma_1, \Gamma'_2 \vdash B:s$ is
derivable. Thus, $\Gamma_1, \Gamma'_2, y:B$ is well-formed.
\end{proof}

If $\Gamma_1$ and $\Gamma_2$ are two well-formed contexts with no variables
in common, then the concatenation $\Gamma_1, \Gamma_2$ also is well-formed.
This remark extend to the case where $\Gamma_1$ and $\Gamma_2$ have variables in common, but are compatible.

\begin{lemma}[Merging]\label{merge}
  If $\Gamma_1$ and $\Gamma_2$ are two well-formed compatible
  contexts, then there exists a well-formed context $\Gamma$, such
  that $\Gamma_1 \subseteq \Gamma$, $\Gamma_2 \subseteq \Gamma$, and
  $\Gamma \subseteq (\Gamma_1, \Gamma_2)$.
\end{lemma}

\begin{proof}
By induction on of $\Gamma_2$.
\begin{itemize}
\item If $\Gamma_2$ is empty, we take $\Gamma = \Gamma_1$.  The
  context $\Gamma$ is well-formed, $\Gamma_1 \subseteq \Gamma$,
  $\Gamma_2 \subseteq \Gamma$, and $\Gamma \subseteq (\Gamma_1,
  \Gamma_2)$.

\item If $\Gamma_2 = (\Gamma'_2, x:A)$, then $\Gamma'_2$ is
  well-formed and, by induction hypothesis, there exists a well-formed
  context $\Gamma'$, such that $\Gamma_1 \subseteq \Gamma'$,
  $\Gamma'_2 \subseteq \Gamma'$, and $\Gamma' \subseteq
  (\Gamma_1,\Gamma'_2)$.
\begin{itemize}
\item If $x:A \in \Gamma'$, then we take $\Gamma = \Gamma'$.  The
  context $\Gamma$ is well-formed, $\Gamma_1 \subseteq \Gamma$,
  $\Gamma_2 \subseteq \Gamma$, and $\Gamma \subseteq (\Gamma_1,
  \Gamma_2)$.

\item Otherwise, as $\Gamma_1$ and $\Gamma_2$ are compatible,
  $\Gamma'$ contains no other declaration of $x$. We take $\Gamma =
  \Gamma', x:A$.  We have $\Gamma_1 \subseteq \Gamma$, $\Gamma_2
  \subseteq \Gamma$, $\Gamma \subseteq (\Gamma_1, \Gamma_2)$.  By
  Lemma \ref{thinningT}, as $\Gamma'_2 \vdash A:s$, and $\Gamma'_2
  \subseteq \Gamma'$, and $\Gamma'$ is well-formed, $\Gamma' \vdash
  A:s$ is derivable , thus $\Gamma$ is well-formed.
\end{itemize}
\end{itemize}
\end{proof}

\begin{example}
Consider two sorts $*$ and $\Box$ and an axiom $*:\Box$.  If $\Gamma_1
= \nat:*, \bool:*, z:\nat$ and $\Gamma_2 = \bool:*, \true:\bool,
\nat:*$, the context $\Gamma$ is $\nat:*, \bool:*, z:\nat,
\true:\bool$.
\end{example}

\section{Arbitrary Contexts}

\begin{definition}[The type system ${\cal T}'$]
The system ${\cal T}'$ is formed with the rules of Figure
\ref{typingrules}.  With respect to the system ${\cal T}$, the rule
(sort) is replaced with the rule (sort'), the rule (start) is replaced
with the rule (var'), and the rule (weak) is dropped.
\end{definition}

\begin{example}
Consider two sorts $*$ and $\Box$ and an axiom $*:\Box$.  The judgement
$\nat:*, z:\nat \vdash z:\nat$ is derivable in ${\cal T}'$. So are the judgements
$z:\nat, \nat:* \vdash z:\nat$ and $\nat:*, x:(*~*), z:\nat \vdash
  z:\nat$.
\end{example}

\begin{lemma}[Thinning]
\label{thinningT'}
If $\Gamma$ and $\Gamma'$ are two contexts, such that $\Gamma \subseteq
\Gamma'$ and $\Gamma \vdash t:A$ is derivable in ${\cal T}'$, then
$\Gamma' \vdash t:A$ is derivable in ${\cal T}'$. 
\end{lemma}

\begin{proof}
By induction on the derivation of $\Gamma \vdash t:A$ in ${\cal T}'$. 
\end{proof}

\begin{lemma}[Key lemma]\label{key}
If $\Gamma$ is well-formed and $\Gamma \vdash t:A$ is derivable in
${\cal T}'$, then $\Gamma \vdash t:A$ is derivable in ${\cal T}$.
\end{lemma}

\begin{proof}
By induction on the derivation of $\Gamma \vdash t:A$ in ${\cal T}'$.

\begin{itemize}
\item If the derivation ends with the rule (sort'), then $t = s_1$, $A
  = s_2$, and $\langle s_1, s_2 \rangle \in {\cal A}$.  As $\Gamma$ is
  well-formed, $\Gamma \vdash s_1:s_2$ is derivable in ${\cal T}$ with
  the rules (sort) and (weak).

\item If the derivation ends with the rule (var'), then $t$ is a
  variable $x$, $\Gamma = \Gamma_1, x:A, \Gamma_2$, and $\Gamma \vdash
  A:s$ is derivable in the system ${\cal T}'$.  As $\Gamma$ is
  well-formed, $\Gamma_1 \vdash A:s'$ is derivable in ${\cal
    T}$. Thus, $\Gamma_1, x:A \vdash x:A$ is derivable in ${\cal T}$
  with the rule (start). And, as $\Gamma$ is well-formed, $\Gamma_1,
  x:A, \Gamma_2 \vdash x:A$ is derivable with the rule (weak).

\item If the derivation ends with the rule (prod), then $t = (x:C) \ra
  D$, $A = s_3$, $\Gamma \vdash C:s_1$ is derivable in ${\cal T}'$,
  $\Gamma, x:C \vdash D:s_2$ is derivable in ${\cal T}'$, and $\langle
  s_1, s_2, s_3 \rangle \in {\cal R}$.  Then, as $\Gamma$ is
  well-formed, by induction hypothesis, $\Gamma \vdash C:s_1$ is
  derivable in ${\cal T}$.  Thus, $\Gamma, x:C$ is well-formed and, by
  induction hypothesis again, $\Gamma, x:C \vdash D:s_2$ is derivable
  in ${\cal T}$. So, $\Gamma \vdash (x:C) \ra D:s_3$ is derivable in
  ${\cal T}$ with the rule (prod).

\item If the derivation ends with the rule (abs), then $t = \lambda
  x:C~u$, $A = (x:C) \ra D$, $\Gamma \vdash C:s_1$ is derivable in
  ${\cal T}'$, $\Gamma, x:C \vdash D:s_2$ is derivable in ${\cal T}'$,
  $\Gamma, x:C \vdash u:D$ is derivable in ${\cal T}'$, and $\langle
  s_1, s_2, s_3 \rangle \in {\cal R}$.  By induction hypothesis,
  $\Gamma \vdash C:s_1$ is derivable in ${\cal T}$.  Thus, $\Gamma,
  x:C$ is well-formed and, by induction hypothesis again, $\Gamma, x:C
  \vdash D:s_2$ is derivable in ${\cal T}$ and $\Gamma, x:C \vdash
  u:D$ is derivable in ${\cal T}$. So, $\Gamma \vdash \lambda
  x:C~u:(x:C) \ra D$ is derivable in ${\cal T}$ with the rule (abs).

\item If the derivation ends with the rule (app), then $t = u~v$, $A =
  (v/x)D$, $\Gamma \vdash u:(x:C) \ra D$ is derivable in ${\cal T}'$
  and $\Gamma \vdash v:C$ is derivable in ${\cal T}'$.  By induction
  hypothesis $\Gamma \vdash u:(x:C) \ra D$ is derivable in ${\cal T}$
  and $\Gamma \vdash v:C$ is derivable in ${\cal T}$. Hence $\Gamma
  \vdash u~v:(v/x)D$ is derivable in ${\cal T}$, with the rule (app).

\item If the derivation ends with the rule (conv), then $\Gamma \vdash
  t:C$ is derivable in ${\cal T}'$, $\Gamma \vdash A:s$ is derivable
  in ${\cal T}'$, and $C \equiv A$. By induction hypothesis, $\Gamma
  \vdash t:C$ is derivable in ${\cal T}$ and $\Gamma \vdash A:s$ is
  derivable in ${\cal T}$.  Thus, $\Gamma \vdash t:A$ is derivable in
  ${\cal T}$, with the rule (conv).
\end{itemize}
\end{proof}

\begin{lemma}[Reordering]\label{reordering}
Let $\Gamma$ be a context, $x$ a variable that does not occur in $\Gamma$,
and $\Gamma'$ a well-formed context, such that $\Gamma' \subseteq
(\Gamma, x:C)$ and $\Gamma' \vdash t:A$ is derivable in ${\cal
  T}'$. Then, there exists a well-formed context $\Gamma''$, such that
$\Gamma'', x:C \vdash t:A$ is derivable in ${\cal T}'$.
\end{lemma}

\begin{proof}
  If $x:C$ is in $\Gamma'$, then we have $\Gamma' = \Gamma'_1, x:C,
\Gamma'_2$, and as $\Gamma'_2 \subseteq \Gamma$, $x$ does not occur in
$\Gamma'_2$.  We take $\Gamma'' = \Gamma'_1, \Gamma'_2$.  By Lemma
\ref{strengtheningcontexts}, $\Gamma''$ is well-formed and, by Lemma
\ref{thinningT'}, $\Gamma'', x:C \vdash t:A$ is derivable in ${\cal
  T}'$.

Otherwise, we take $\Gamma'' = \Gamma'$. This context
  is well-formed and, by Lemma \ref{thinningT'},
$\Gamma'', x:C \vdash t:A$ is derivable in ${\cal T}'$.

\end{proof}

\begin{lemma}[Context curation]
\label{curation}
If $\Gamma \vdash t:A$ is derivable in ${\cal T}'$, then there exists
a well-formed context $\Delta$, such that $\Delta \subseteq \Gamma$
and $\Delta \vdash t:A$ is derivable in ${\cal T}'$.
\end{lemma}

\begin{proof}
By induction on the derivation of $\Gamma \vdash t:A$.
\begin{itemize}
\item If the derivation ends with the rule (sort'), then $t = s_1$ and
  $A = s_2$, such that $\langle s_1, s_2 \rangle \in {\cal A}$.  We
  take the empty context for $\Delta$, $\Delta \subseteq \Gamma$,
  $\Delta$ is well-formed, and $\Delta \vdash s_1:s_2$ is derivable in
  ${\cal T}'$, with the rule (sort').

\item If the derivation ends with the rule (var'), then $t$ is a
  variable $x$, $x:A$ is an element of $\Gamma$ and $\Gamma \vdash
  A:s$ is derivable in ${\cal T}'$.  By induction hypothesis, there
  exists a well-formed context $\Delta_1$, such that $\Delta_1
  \subseteq \Gamma$ and $\Delta_1 \vdash A:s$ is derivable in ${\cal
    T}'$.
  
If $x:A$ is an element of $\Delta_1$, we take $\Delta = \Delta_1$.  We
have $\Delta \subseteq \Gamma$ and $\Delta$ is well-formed. Moreover
$\Delta \vdash A:s$ is derivable in ${\cal T}'$ and $\Delta$ contains
$x:A$, thus $\Delta \vdash x:A$ is derivable in ${\cal T}'$, with the
rule (var').

Otherwise, as $\Delta_1 \subseteq \Gamma$, $\Delta_1$ contains no
declaration of $x$, we take $\Delta = \Delta_1, x:A$.  We have $\Delta
\subseteq \Gamma$.  By Lemma \ref{key}, $\Delta_1 \vdash A:s$ is
derivable in ${\cal T}$, thus $\Delta$ is well-formed.  Moreover, by
Lemma \ref{thinningT'}, the judgement $\Delta \vdash A:s$ is derivable
in ${\cal T}'$ and, as $\Delta$ contains $x:A$, $\Delta \vdash x:A$ is
derivable in ${\cal T}'$, with the rule (var').

\item If the derivation ends with the rule (prod) then $t = (x:C) \ra
  D$, $A = s_3$, the contexts $\Gamma \vdash C:s_1$ and $\Gamma, x:C
  \vdash D:s_2$ are derivable in ${\cal T}'$, and $\langle s_1, s_2,
  s_3 \rangle \in {\cal R}$.  Modulo $\alpha$-equivalence, we can
  assume that $x$ does not occur in $\Gamma$.  By induction hypothesis,
  there exist two well-formed contexts $\Gamma_1$ and $\Gamma_2$, such
  that $\Gamma_1 \subseteq \Gamma$, $\Gamma_2 \subseteq (\Gamma,
  x:C)$, and the judgements $\Gamma_1 \vdash C:s_1$ and $\Gamma_2
  \vdash D:s_2$ are derivable in ${\cal T}'$.

By Lemma \ref{reordering}, there exists a well-formed context
$\Gamma'_2$ such that
$\Gamma'_2, x:C \vdash D:s_2$ is derivable in
${\cal T}'$.
As $\Gamma_1$ and $\Gamma'_2$ contain no declaration of $x$, 
by Lemma \ref{merge}, there exists a well-formed context $\Delta$,
such that $\Gamma_1 \subseteq \Delta$,
$\Gamma'_2 \subseteq \Delta$, and $\Delta$ contains no declaration of $x$.
We have $\Gamma_1 \subseteq \Delta$ and
$\Gamma'_2,x:C \subseteq \Delta,x:C$.  Thus, by Lemma
\ref{thinningT'}, $\Delta \vdash C:s_1$ and $\Delta, x:C \vdash D:s_2$
are derivable in ${\cal T}'$.  Thus, $\Delta \vdash (x:C) \ra D:s_3$
is derivable in ${\cal T}'$, with the rule (prod).

\item If the derivation ends with the rule (abs), then $t =
  \lambda x:C~u$, $A = (x:C) \ra D$, the judgements $\Gamma \vdash
  C:s_1$, $\Gamma, x:C \vdash D:s_2$, and $\Gamma, x:C \vdash u:D$ are
  derivable in ${\cal T}'$, and $\langle s_1, s_2, s_3 \rangle \in
  {\cal R}$.  Modulo $\alpha$-equivalence, we can assume that $x$ does
  not occur in $\Gamma$.  By induction hypothesis, there exist three
  well-formed contexts $\Gamma_1$, $\Gamma_2$, and $\Gamma_3$, such
  that $\Gamma_1 \subseteq \Gamma$, $\Gamma_2 \subseteq (\Gamma,x:C)$,
  $\Gamma_3 \subseteq (\Gamma,x:C)$, and the judgements $\Gamma_1
  \vdash C:s_1$, $\Gamma_2 \vdash D:s_2$, and $\Gamma_3 \vdash u:D$
  are derivable in ${\cal T}'$.

By Lemma \ref{reordering}, there exists well-formed contexts
$\Gamma'_2$ and $\Gamma'_3$, such that the judgements $\Gamma'_2, x:C
\vdash D:s_2$ and $\Gamma'_3, x:C \vdash u:D$ are derivable in ${\cal
  T}'$.  As $\Gamma_1$, $\Gamma'_2$, and $\Gamma'_3$ contain no
declaration of $x$, using Lemma \ref{merge} twice, there exists a
well-formed context $\Delta$, such that $\Gamma_1 \subseteq \Delta$,
$\Gamma'_2 \subseteq \Delta$, $\Gamma'_3 \subseteq \Delta$, and
$\Delta$ contains no declaration of $x$.  We have $\Gamma_1 \subseteq
\Delta$, $\Gamma'_2,x:C \subseteq \Delta,x:C$, and $\Gamma'_3,x:C
\subseteq \Delta,x:C$. Thus, by Lemma \ref{thinningT'}, the judgements
$\Delta \vdash C:s_1$, $\Delta,x:C \vdash D:s_2$, and $\Delta,x:C
\vdash u:D$ are derivable in ${\cal T}'$.  Thus, $\Delta \vdash
\lambda x:C~u:(x:C) \ra D$ is derivable in ${\cal T}'$, with the
rule (abs).

\item If the derivation ends with the rule (app) then $t = u~v$,
  $A = (v/x)D$, and the judgements $\Gamma \vdash u:(x:C) \ra D$ and
  $\Gamma \vdash v:C$ are derivable in ${\cal T}'$.  By induction
  hypothesis, there exists two well-formed contexts $\Gamma_1$ and
  $\Gamma_2$, such that $\Gamma_1 \subseteq \Gamma$, $\Gamma_2
  \subseteq \Gamma$, and the judgements $\Gamma_1 \vdash u:(x:C) \ra
  D$ and $\Gamma_2 \vdash v:C$ are derivable in ${\cal T}'$.

By Lemma \ref{merge}, there exists a well-formed context $\Delta$,
such that $\Gamma_1 \subseteq \Delta$, and $\Gamma_2 \subseteq
\Delta$.  By Lemma \ref{thinningT'}, the judgements $\Delta \vdash
u:(x:C) \ra D$ and $\Delta \vdash v:C$ are derivable in ${\cal T}'$.
Thus, $\Delta \vdash (u~v):(v/x)D$ is derivable in ${\cal T}'$, with
the rule (app).

\item If the derivation ends with the rule (conv) then the
  judgements $\Gamma \vdash t:C$ and $\Gamma \vdash A:s$ are derivable
  in ${\cal T}'$, and $C \equiv A$.  By induction hypothesis, there
  exists two well-formed contexts $\Gamma_1$ and $\Gamma_2$, such that
  $\Gamma_1 \subseteq \Gamma$, $\Gamma_2 \subseteq \Gamma$, and the
  judgements $\Gamma_1 \vdash t:C$ and $\Gamma_2 \vdash A:s$ are
  derivable in ${\cal T}'$.

By Lemma \ref{merge}, there exists a well-formed context $\Delta$,
such that $\Gamma_1 \subseteq \Delta$ and $\Gamma_2 \subseteq \Delta$.
By Lemma \ref{thinningT'}, the judgements $\Delta \vdash t:C$ and
$\Delta \vdash A:s$ are derivable in ${\cal T}'$.  Thus, $\Delta
\vdash t:A$ is derivable in ${\cal T}'$, with the rule (conv).
\end{itemize}
\end{proof}

\begin{theorem}\label{maintheorem}
If $\Gamma \vdash t:A$ is derivable in ${\cal T}'$, then there exists
$\Delta$, such that $\Delta \subseteq \Gamma$ and $\Delta \vdash t:A$
is derivable in ${\cal T}$.
\end{theorem}

\begin{proof}
By Lemma \ref{curation}, there exists a well-formed context $\Delta$,
such that $\Delta \subseteq \Gamma$ and $\Delta \vdash t:A$ is
derivable in ${\cal T}'$.  By Lemma \ref{key}, $\Delta \vdash t:A$ is
derivable in ${\cal T}$.
\end{proof}

\begin{example}
Consider two sorts $*$ and $\Box$ and an axiom $*:\Box$.  
From the derivation of the judgement,
$z:\nat, \nat:* \vdash z:\nat$, 
we extract the context $\nat:*, z:\nat$. 

And from the derivation of $\nat:*, x:(*~*), z:\nat \vdash z:\nat$, we
also extract the context $\nat:*, z:\nat$.
\end{example}

\section*{Acknowledgements}

The author wants to thank Fr\'ed\'eric Blanqui, Herman Geuvers, and
Claudio Sacerdoti Coen for useful and lively discussions about the
various presentations of type theory.

\bibliographystyle{eptcs}
\bibliography{unsafe}

\begin{thebibliography}{1}
\providecommand{\bibitemdeclare}[2]{}
\providecommand{\surnamestart}{}
\providecommand{\surnameend}{}
\providecommand{\urlprefix}{Available at }
\providecommand{\url}[1]{\texttt{#1}}
\providecommand{\href}[2]{\texttt{#2}}
\providecommand{\urlalt}[2]{\href{#1}{#2}}
\providecommand{\doi}[1]{doi:\urlalt{http://dx.doi.org/#1}{#1}}
\providecommand{\bibinfo}[2]{#2}

\bibitemdeclare{incollection}{Barendregt1992}
\bibitem{Barendregt1992}
\bibinfo{author}{H.~\surnamestart Barendregt\surnameend}
  (\bibinfo{year}{1992}): \emph{\bibinfo{title}{Lambda calculi with types}}.
\newblock In \bibinfo{editor}{S.~\surnamestart Abramsky\surnameend},
  \bibinfo{editor}{D.M. \surnamestart Gabbay\surnameend} \&
  \bibinfo{editor}{T.S.E. \surnamestart Maibaum\surnameend}, editors: {\sl
  \bibinfo{booktitle}{Handbook of {Logic} in {Computer} {Science}}},
  \bibinfo{volume}{2}, \bibinfo{publisher}{Oxford University Press}, pp.
  \bibinfo{pages}{117--309}.

\bibitemdeclare{inproceedings}{Blanqui}
\bibitem{Blanqui}
\bibinfo{author}{F.~\surnamestart Blanqui\surnameend} (\bibinfo{year}{2001}):
  \emph{\bibinfo{title}{Definitions by Rewriting in the Calculus of
  Constructions}}.
\newblock In: {\sl \bibinfo{booktitle}{Logic in Computer Science}},
  \bibinfo{publisher}{{IEEE} Computer Society}, pp. \bibinfo{pages}{9--18},
  \doi{10.1109/LICS.2001.932478}.

\bibitemdeclare{inproceedings}{Sacerdoti}
\bibitem{Sacerdoti}
\bibinfo{author}{C.~Sacerdoti \surnamestart Coen\surnameend}
  (\bibinfo{year}{2004}): \emph{\bibinfo{title}{Mathematical Libraries as Proof
  Assistant Environments}}.
\newblock In \bibinfo{editor}{A.~\surnamestart Asperti\surnameend},
  \bibinfo{editor}{G.~\surnamestart Bancerek\surnameend} \&
  \bibinfo{editor}{A.~\surnamestart Trybulec\surnameend}, editors: {\sl
  \bibinfo{booktitle}{Mathematical Knowledge Management}}, {\sl
  \bibinfo{series}{Lecture Notes in Computer Science}} \bibinfo{volume}{3119},
  \bibinfo{publisher}{Springer}, pp. \bibinfo{pages}{332--346},
  \doi{10.1007/978-3-540-27818-4\_24}.

\bibitemdeclare{article}{CoquandHuet}
\bibitem{CoquandHuet}
\bibinfo{author}{T.~\surnamestart Coquand\surnameend} \&
  \bibinfo{author}{G.~\surnamestart Huet\surnameend} (\bibinfo{year}{1988}):
  \emph{\bibinfo{title}{The Calculus of Constructions}}.
\newblock {\sl \bibinfo{journal}{Information and Computation}}
  \bibinfo{volume}{76}(\bibinfo{number}{2/3}), pp. \bibinfo{pages}{95--120},
  \doi{10.1016/0890-5401(88)90005-3}.

\bibitemdeclare{inproceedings}{CD}
\bibitem{CD}
\bibinfo{author}{D.~\surnamestart Cousineau\surnameend} \&
  \bibinfo{author}{G.~\surnamestart Dowek\surnameend} (\bibinfo{year}{2007}):
  \emph{\bibinfo{title}{Embedding Pure Type Systems in the Lambda-Pi-Calculus
  Modulo}}.
\newblock In \bibinfo{editor}{S.~Ronchi~Della \surnamestart Rocca\surnameend},
  editor: {\sl \bibinfo{booktitle}{Typed Lambda Calculi and Applications}},
  {\sl \bibinfo{series}{Lecture Notes in Computer Science}}
  \bibinfo{volume}{4583}, \bibinfo{publisher}{Springer}, pp.
  \bibinfo{pages}{102--117}, \doi{10.1007/978-3-540-73228-0\_9}.

\bibitemdeclare{inproceedings}{GeuversKrebbersMcKinnaWiedijk}
\bibitem{GeuversKrebbersMcKinnaWiedijk}
\bibinfo{author}{H.~\surnamestart Geuvers\surnameend},
  \bibinfo{author}{R.~\surnamestart Krebbers\surnameend},
  \bibinfo{author}{J.~\surnamestart McKinna\surnameend} \&
  \bibinfo{author}{F.~\surnamestart Wiedijk\surnameend} (\bibinfo{year}{2010}):
  \emph{\bibinfo{title}{Pure Type Systems without Explicit Contexts}}.
\newblock In \bibinfo{editor}{K.~\surnamestart Crary\surnameend} \&
  \bibinfo{editor}{M.~\surnamestart Miculan\surnameend}, editors: {\sl
  \bibinfo{booktitle}{Logical Frameworks and Meta-languages: Theory and
  Practice}}, {\sl \bibinfo{series}{Electronic Proceedings in Theoretical
  Computer Science}}~\bibinfo{volume}{34}, \bibinfo{publisher}{Open Publishing
  Association}, pp. \bibinfo{pages}{53--67}, \doi{10.4204/EPTCS.34.6}.

\bibitemdeclare{article}{GeuversNederhof}
\bibitem{GeuversNederhof}
\bibinfo{author}{H.~\surnamestart Geuvers\surnameend} \& \bibinfo{author}{M.-J.
  \surnamestart Nederhof\surnameend} (\bibinfo{year}{1991}):
  \emph{\bibinfo{title}{Modular proof of strong normalization for the calculus
  of constructions}}.
\newblock {\sl \bibinfo{journal}{Journal of Functional Programming}}
  \bibinfo{volume}{1}(\bibinfo{number}{2}), pp. \bibinfo{pages}{155--189},
  \doi{10.1017/S0956796800020037}.

\end{thebibliography}
\end{document}